%% file: main.tex
\documentclass[final]{siamart220329}
\usepackage{braket,amsfonts}

\usepackage{array}

\usepackage[caption=false]{subfig}
\usepackage{pgfplots}

\newsiamthm{claim}{Claim}
\newsiamremark{remark}{Remark}
\newsiamremark{example}{Example}
\newsiamremark{hypothesis}{Hypothesis}
\crefname{hypothesis}{Hypothesis}{Hypotheses}

\usepackage{algorithmic}

\usepackage{graphicx,epstopdf}

\Crefname{ALC@unique}{Line}{Lines}

\usepackage{amsopn}

\usepackage{xspace}
\usepackage{bold-extra}
\usepackage[most]{tcolorbox}

\colorlet{texcscolor}{blue!50!black}
\colorlet{texemcolor}{red!70!black}
\colorlet{texpreamble}{red!70!black}
\colorlet{codebackground}{black!25!white!25}

\usepackage{mathtools}
\usepackage{amssymb}

\newcommand{\range}[1]{[1, #1]}
\newcommand{\frange}[2]{[#1, #2]}
\newcommand{\B}{\mathbb B}
\newcommand{\C}{\{0,1,\ast\}}

\newcommand{\vb}[1]{\mathbf{#1}}

\newcommand{\vertices}{v}
\newcommand{\TT}{{t}}

\newcommand{\T}{{\bf T}}

\lstdefinestyle{siamlatex}{%
  style=tcblatex,
  texcsstyle=*\color{texcscolor},
  texcsstyle=[2]\color{texemcolor},
  keywordstyle=[2]\color{texemcolor},
  moretexcs={cref,Cref,maketitle,mathcal,text,headers,email,url},
}

\tcbset{%
  colframe=black!75!white!75,
  coltitle=white,
  colback=codebackground, 
  colbacklower=white, 
  fonttitle=\bfseries,
  arc=0pt,outer arc=0pt,
  top=1pt,bottom=1pt,left=1mm,right=1mm,middle=1mm,boxsep=1mm,
  leftrule=0.3mm,rightrule=0.3mm,toprule=0.3mm,bottomrule=0.3mm,
  listing options={style=siamlatex}
}

\newtcbinputlisting[use counter=example]{\examplefile}[3][]{%
  title={Example~\thetcbcounter: #2},listing file={#3},#1}

\DeclareTotalTCBox{\code}{ v O{} }
{ 
  fontupper=\ttfamily\color{black},
  nobeforeafter,
  tcbox raise base,
  colback=codebackground,colframe=white,
  top=0pt,bottom=0pt,left=0mm,right=0mm,
  leftrule=0pt,rightrule=0pt,toprule=0mm,bottomrule=0mm,
  boxsep=0.5mm,
  #2}{#1}

\patchcmd\newpage{\vfil}{}{}{}
\begin{tcbverbatimwrite}{tmp_\jobname_header.tex}
\title{Computational Complexity of Minimal Trap Spaces in Boolean Networks\thanks{%
\funding{KM and KL acknowledges support from the National Research Foundation of Korea(NRF) grant
	funded by the Korea government(MEST) (grant number: NRF-2022R1F1A1074140)
	LP acknowledges support from the French Agence Nationale pour la Recherche (ANR) in the scope of the project BNeDiction (grant number: ANR-20-CE45-0001).}}}

\author{Kyungduk Moon\thanks{Department of Industrial and Management Engineering, POSTECH, South Korea (\email{kaleb.moon@postech.ac.kr}, \email{kblee@postech.ac.kr}).}
\and Kangbok Lee\footnotemark[2]
\and Loïc Paulevé\thanks{Univ. Bordeaux, CNRS, Bordeaux INP, LaBRI, UMR 5800 F-33400 Talence, France (\email{loic.pauleve@labri.fr})}}

\headers{Computational complexity of trap spaces in Boolean networks}{Kyungduk Moon, Kangbok Lee, and Loïc Paulevé}
\end{tcbverbatimwrite}
\input{tmp_\jobname_header.tex}

\hypersetup{ pdftitle={Computational complexity of trap spaces in Boolean networks} }

\begin{document}
\maketitle

\begin{tcbverbatimwrite}{tmp_\jobname_abstract.tex}
\begin{abstract}
A Boolean network (BN) is a discrete dynamical system defined by a Boolean function that maps to the domain itself. A trap space of a BN is a generalization of a fixed point, which is defined as the sub-hypercubes closed by the function of the BN. A trap space is minimal if it does not contain any smaller trap space.
Minimal trap spaces have applications for the analysis of attractors of BNs with various update modes.
This paper establishes the computational complexity results of three decision problems related to minimal trap spaces:
the decision of the trap space property of a sub-hypercube,
the decision of its minimality, and 
the decision of the membership of a given configuration to a minimal trap space.
Under several cases on Boolean function 
representations, we investigate the computational complexity
of each problem.
In the general case, we demonstrate that the trap space property is coNP-complete, and the minimality
and the membership properties are $\Pi_2^{\text P}$-complete.
The complexities drop by one level in the polynomial hierarchy whenever the local functions of the
BN are either unate, or are 
represented using truth-tables, binary decision diagrams, or double DNFs (Petri net encoding):
the trap space property can be decided in a polynomial time, whereas deciding the minimality and the membership are coNP-complete.
When the BN is given as its functional graph, all these problems are in P.
\end{abstract}


\end{tcbverbatimwrite}
\input{tmp_\jobname_abstract.tex}

\section{Introduction} \label{sec:intro}

A Boolean network (BN) is a dynamical system defined by 
a function $f$ of the Boolean domain with a fixed dimension $n$ that maps to the domain itself,
i.e., $f: \B^n\to\B^n$ with $\B=\{0,1\}$. 
The function mapping to a component of the image of $f$ is called a \emph{local function}. We denote the local function mapping to the $i$-th component of the image as $f_i:\B^n\to\B$ for $i \in \{1, \ldots, n\}$.
Given a Boolean vector $\vb x\in\B^n$ 
referred to as a \emph{configuration}, one can define a set
of succeeding configurations by $f$ following an \emph{update mode}~\cite{aracena09,Kauffman69,PS21,thomas73}, leading to a dynamical system. 
Two popular update modes are the synchronous update mode and  asynchronous update modes.
The \textit{synchronous update mode} associates $f(\vb x)$ as the unique succeeding configuration of $\vb x$ by $f$. On the other hand, an asynchronous update mode may associate multiple configurations of which some components match to the corresponding local functions evaluated with $\vb x$; if the $i$-th component matches, then ${\vb y}_i = f_i(\vb x)$.
The \textit{fully asynchronous update mode} is a specific example that associates $\vb x$ with any configuration
$\vb y$ which differs from $\vb x$ by exactly one component that matches to the corresponding local function.
BNs are studied in various disciplines such as discrete mathematics~\cite{aracena_2011_Combinatorics,richard_2009_Positive} and 
dynamical system theory~\cite{Bridoux_2020,Gamard21,pauleve_2012_Static}. They also have wide applications to the modeling of complex systems such as biological
systems~\cite{Akutsu2018,kauffman1993origins,Montagud22,schwab_concepts_2020,Zanudo21}, and social
behaviors~\cite{Grabisch_2013,Poindron_2021}, to name but a few.

The literature addresses a vast zoo of update modes, generating possibly different dynamics from the
same BN $f$~\cite{PS2022}.
In this context, the dynamical properties of BNs which are independent of the update mode are of particular interest because they show inherent dynamical properties of a given BN.
The prime example is the study of fixed points of $f$, i.e., the configurations $\vb x$ such that
$f(\vb x)=\vb x$.
Indeed, a fixed point of $f$ can be assumed to be a stable state of dynamics that does not change after a transition with any update mode. 
Conversely, a stable state under the synchronous or the asynchronous update mode is a fixed point.
Nevertheless, some specific update modes may exhibit additional stable states~\cite{PS2022}.
The notion of fixed points of $f$ can be generalized to \emph{trap spaces}.
A trap space is a sub-hypercube (an $n$-dimensional sub-graph of the $n$-dimensional hypercube where some dimensions can
be  fixed to be singular values)
closed by $f$ such that for any vertex $\vb x$ of the trap space,
$f(\vb x)$ is also one of its vertices.
A \emph{minimal trap space} is a trap space that contains no other trap spaces.
Therefore, a fixed point of $f$ is a specific case of minimal trap spaces where all dimensions are fixed.

Minimal trap spaces in BNs have been studied as approximations of \emph{attractors} with
(a)synchronous update modes~\cite{klarner2015Approximating}, which are important features
of the long-term dynamical properties of BNs~\cite{kauffman1993origins}.
Given an update mode, an attractor is defined as an inclusion-wise minimal set of configurations
which are closed by transitions. Equivalently, an attractor is a set of configurations
satisfying the following two conditions. First, there exists a sequence of transitions between any
pair of its configurations. Second, if there exists a sequence of transitions from one of its configuration $\vb x$ to another configuration $\vb x'$, then they belong to the same attractor.
If an attractor is composed by a single configuration, it is a fixed point; otherwise, it is called a \textit{cyclic attractor}.
It appears that any minimal trap space necessarily encloses at least one attractor of any update
mode~\cite{klarner2015Computing,pauleve2020Reconciling,PS2022}.
Moreover, if the minimal trap space is not a fixed point, the enclosed (a)synchronous attractors are necessarily
cyclic.
Beside the synchronous and the asynchronous update mode,
minimal trap spaces are exactly the attractors of BNs under the most permissive update
mode~\cite{pauleve2020Reconciling,PS21}, which guarantees to capture all transitions realized by any
multi-valued refinement of the BN.

So far, the literature has essentially focused on algorithms and implementations for enumerating
minimal trap spaces of
BNs~\cite{klarner2015Computing,pauleve2020Reconciling,trinh_2022_Minimal}.
Nevertheless, whereas these algorithms indicate upper bounds for the computational complexity of decision problems related to minimal trap spaces, no lower bound has been characterized.
In this paper, we provide computational complexity results of problems related to the minimal trap spaces.
We focus on three fundamental decisions problems:
\begin{description}
	\item[TRAPSPACE($f$, $\vb h$)]:~\\
	Given a BN $f$ and a sub-hypercube $\vb h$, $\vb h$ is a trap space of $f$.
	\item[MINTRAP($f$, $\vb h$)]:~\\
	Given a BN $f$ and a sub-hypercube $\vb h$, $\vb h$ is a minimal trap space of $f$.
	\item[IN-MINTRAP($f$, $\vb x$)]:~\\
	Given a BN $f$ and a configuration $\vb x$, $\vb x$ is a vertex of a minimal trap space of $f$.
\end{description}

We study the computational complexity of these problems depending on Boolean function
representations and on the unate property of local functions, as summarized in
Table~\ref{tab:results}.
In the case of Boolean functions represented as propositional formulas,
upper bounds of these three decisions problems have been determined in
\cite{pauleve2020Reconciling}:
TRAPSPACE is in coNP, whereas MINTRAP and IN-MINTRAP are in $\Pi^{\text P}_2$.
Moreover, whenever the BN is \emph{locally monotone}, i.e., each of its local function is unate (its
expression does not contain a variable appearing both positively and negatively),
they showed that
TRAPSPACE is in P, whereas MINTRAP and IN-MINTRAP are in coNP.
We complete the results of \cite{pauleve2020Reconciling} by demonstrating the lower bound results
for each corresponding case. We further consider three representations for local functions: truth
tables, binary decision diagrams, and double DNFs. These representations have practical relevance
since binary decision diagrams are frequently employed by software  \cite{Naldi_2018} and double
DNFs are employed by Petri nets \cite{CHJPS14-CMSB,gored-TCBB,trinh_2022_Minimal}.
For all three representations, the same computational complexities are demonstrated as the locally monotone case.
Finally, we also consider a BN
represented by its \textit{functional graph}, which matches with the state transition graph
under the synchronous update mode:
this graph associates each state in $\B^n$ to its image by $f$.
Therefore, we consider two classes of representations of BNs: either by the representation of their local functions, or by the representation of the global function $f$.

\begin{table}[ht]
	\caption{\label{tab:results} The computational complexity of decision problems related to
		trap spaces}
	\centering
	\begin{tabular}{@{}cc c ccc@{}}
		\hline
		\multicolumn{2}{c}{\textbf{Boolean network}} && \multicolumn{3}{c}{\textbf{Problems}} \\
		\cline{1-2} \cline{4-6} 
		Representation & Unate property &&  TRAPSPACE  &  MINTRAP  &  IN-MINTRAP \\ 		\hline
		\multicolumn{5}{@{}l}{\textbf{\textit{
					Representation with local functions
		}}} \\ 
		Propositional formula & General &&  coNP-complete$^\dagger$  &  $\Pi_2^{\mathrm P}$-complete$^\dagger$  &  $\Pi_2^{\mathrm P}$-complete$^\dagger$\\
		Propositional formula & Locally monotone  &&  P$^\dagger$  &  coNP-complete$^\dagger$  &  coNP-complete$^\dagger$  \\
		Truth table & General &&  P  &  coNP-complete   &  coNP-complete \\
		Binary decision diagram & General &&  P  &  coNP-complete   &  coNP-complete \\
		Double DNFs & General &&  P  &  coNP-complete   &  coNP-complete \\[0.3em]
		\multicolumn{6}{@{}l}{\textbf{\textit{
					Representation with a state transition graph
		}}}
		\\[0.3em]
		Functional graph   & General &&  P  &  P  &  P \\
				\hline
		\multicolumn{6}{l}{
			\begin{minipage}{0.9\textwidth}
				\small $^\dagger$the upper bound results are presented in \cite{pauleve2020Reconciling}
			\end{minipage}
		}
	\end{tabular}
\end{table}

The rest of this paper is organized as follows.
In Sec. \ref{sec:definition}, we introduce notations and terminologies that are used to define the problems and explain the results. 
In Sec. \ref{sec:main_results}, we provide computational complexity results of these problems in
different BN settings.
Sec. \ref{sec:conclusion} provides a concluding remark.

\section{Preliminaries}\label{sec:definition}

We denote integers ranging from $1$ to $n$ by $\range n:=\{1, \ldots, n\}$.
We use an interval subscript of a vector to denote the list of components lying on the interval's range. For example, $\vb x_{[1,n_1]}$ denotes the vector concatenating the first $n_1$ components of vector $\vb x$. 

\subsection{Representations and the unate property of local functions in Boolean networks}
\label{subsec:notation_representation}

A Boolean function on $n$ variables is of the form $\phi:\B^n\to \B$.
In this context, the following decision problems are relevant to our study.
The SAT problem is to decide  whether there exists a configuration 
$\vb x\in\B^n$ such that $\phi(\vb x)=1$.
The TAUTOLOGY problem is to decide if 
$\phi(\vb x)=1$ for all configuration $\vb x\in\B^n$.
The NOT-TAUTOLOGY problem is to decide if there exists a configuration $\vb x$ such that
$\phi(\vb x)=0$.
We can generalize SAT by partitioning variables and alternately putting quantifiers `$\exists$' and `$\forall$'  on them. Given the number of quantifiers $l\in \mathbb{Z}^+$, and the number of indices of partitions $n_1, \ldots, n_l$ such that $n_{j} < n_{j+1}$ for each $j\in\range {l-1}$, we can define the following generalization to $\Sigma_l$SAT:
\begin{align*}
	\Sigma_l\text{SAT: decide if } 
	\exists \vb x_{[1, n_1]} 
	\forall \vb x_{[n_1+1, n_2]}
	\exists \vb x_{[n_2+1, n_3]} 
	\cdots 
	\phi(\vb x)=1.
\end{align*}
Analogously, we can generalize TAUTOLOGY to $\Pi_l$SAT by replacing `$\exists$' with `$\forall$' and vice versa from $\Sigma_l$SAT. Notably, $\Sigma_l$SAT and $\Pi_l$SAT are together called 
\textit{the true quantified Boolean formula} (TQBF) problem, which is well known as a complete
problem in the polynomial hierarchy~\cite{stockmeyer1976Polynomialtime}.
The polynomial hierarchy is a generalization of NP and coNP that can be defined with \textit{oracles}, imaginary machines that instantly give an answer to a group of decision problems~\cite{arora_2009_Computational}. Given an oracle for NP, we can define $\Sigma^\text{P}_2$ as the group of problems of which the answer is verified as true in a nondeterministic polynomial time. Analogously, we can define $\Pi^\text{P}_2$ as the group of problems of which the answer is verified as false in a nondeterministic polynomial time using an oracle for NP. Furthermore, we can define $\text{P}^\text{NP}$ as the group of problems that are  polynomial time solvable using an oracle for NP. Note that $\Sigma^\text{P}_2$, $\Pi^\text{P}_2$, and $\text{P}^\text{NP}$ can be equivalently defined using an oracle for coNP with no difference. We can repeat this to define the classes of problems in the polynomial hierarchy, as depicted in  Figure \ref{fig:hierarchy}  where the arrows denote inclusion of the classes. For each class, we can define complete problems that can be reduced from all problems in the same class, implying that they are the hardest problems in their classes.
In this paper, we limit our focus to SAT, TAUTOLOGY, NOT-TAUTOLOGY, and $\Pi_2$SAT.
The computational complexity of these three problems depends on the representation of $\phi$ and can be
refined depending on its unate property.

\begin{figure}[h]
	\centering
	\begin{tikzpicture}[->, node distance=2cm, semithick]
		\node (P) {$\Sigma_0^\text{P}$ = P = $\Pi_0^\text{P}$};
		\node (Sigma1) [above left of=P]       {NP = $\Sigma_1^\text{P}$ \hspace*{0.9cm}};
		\node (Pi1)    [above right of=P]      {\hspace*{1.2cm} $\Pi_1^\text{P}$ = coNP};
		\node (Delta2) [above left of=Pi1]     {$\text{P}^\text{NP}$};
		\node (Sigma2) [above left of=Delta2]  {$\Sigma_2^\text{P}$};
		\node (Pi2)    [above right of=Delta2] {$\Pi_2^\text{P}$};
		\node (dots)   [above of=Delta2]       {\vdots};
		\draw (P)      -> (Sigma1);
		\draw (P)      -> (Pi1);
		\draw (Sigma1) -> (Sigma2);
		\draw (Sigma1) -> (Delta2);
		\draw (Pi1)    -> (Pi2);
		\draw (Pi1)    -> (Delta2);
		\draw (Delta2) -> (Sigma2);
		\draw (Delta2) -> (Pi2);
	\end{tikzpicture}
	\caption{The first two levels of the polynomial hierarchy}
	\label{fig:hierarchy}
\end{figure}
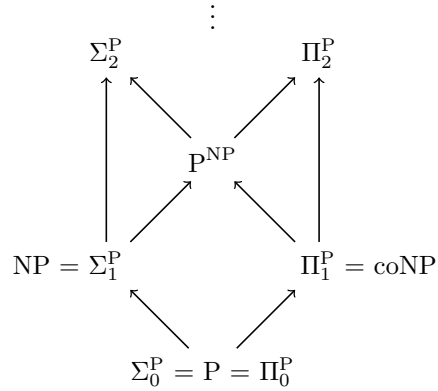

\paragraph{As propositional formula}
Boolean function $\phi$ can be represented as a propositional formula, which  consists of Boolean variables
$\vb x_1, \ldots,\vb  x_n$ and logical connectives $\wedge$ (conjunction), $\vee$ (disjunction), and
$\neg$ (negation).
The size of a formula is its length, which is proportional to the total count of variables and connectives appearing in the input string.
With this representation, SAT and NOT-TAUTOLOGY are NP-complete, TAUTOLOGY is coNP-complete, and $\Pi_2$SAT is $\Pi_2$-complete~\cite{stockmeyer1976Polynomialtime}. 

\paragraph{As disjunctive normal form (DNF)}
Boolean function $\phi$ 
can be represented as a propositional formula consisting of a disjunction of
conjunctive clauses. Negations are allowed only on variables and not on clauses. A DNF can be
equivalently represented as a list of sets of literals, where a literal is either a variable or a
negated variable. Any propositional formula can be represented in DNF with clauses consisting of at
most three literals (3DNF), although its size may be exponentially large.
When $\phi$ is represented as a DNF (or 3DNF), SAT is in P because $\phi$ is satisfiable if
and only if it contains at least one clause and none of the variables appear as both positive and
negative literal in a single clause. TAUTOLOGY is coNP-complete and NOT-TAUTOLOGY is NP-complete
from the fact that 3CNF SAT (given as a conjunctive normal form with three literals per clause) is NP-complete  \cite{arora_2009_Computational}.

\paragraph{Unate (monotone) case}
Boolean function  $\phi$ is \emph{unate} if
there exists an ordering of components
\(\preceq\in \{\leq,\geq\}^n\)
such that
\(
\forall \vb x,\vb y\in\B^n,
\Big(
(\vb x_1\preceq_1 \vb y_1)
\wedge \cdots \wedge 
(\vb x_n\preceq_n \vb y_n)
\Big) \Rightarrow
\phi(\vb x) \leq \phi(\vb y)
\).
In other words, for each component $j\in\range n$ and every configuration $\vb x\in\B^n$, $\phi(\vb x_{[1,j-1]}0\vb x_{[j+1,n]}) \preceq_j \phi(\vb x_{[1,j-1]}1\vb x_{[j+1,n]})$ holds.
When $\phi$ is unate and the ordering of components is provided, SAT, TAUTOLOGY, and NOT-TAUTOLOGY
can be decided in P because only the maximal (or the minimal) assignment to the ordering needs to be
evaluated.

\paragraph{As truth table (TT)}
Boolean function  $\phi$ can be encoded as binary vector $\TT$ with $2^n$ rows, where for each row
$m\in\range{2^n}$, $\TT_{m}$ is the value of $f_i(\operatorname{bin}(m-1))$ with
$\operatorname{bin}(m)$ being the binary representation of $m$.
With the truth table representation, SAT, TAUTOLOGY, and NOT-TAUTOLOGY are in P because a satisfying
or an unsatisfying assignment can be searched in linear time to the size of the truth table.

\paragraph{As binary decision diagram (BDD)}
A BDD has a directed acyclic graph structure with a unique root and at most two terminal nodes
among 0 and 1~\cite{drechsler_1998_Binary}.
Each non-terminal node is associated to a component $i\in\range n$ and has two out-going edges,
one labeled with 0 and the other with 1.
Moreover, any path from the root to a terminal node crosses at most one node associated
to each component.
Then, each configuration $\vb x$ corresponds to a single path from the root to a terminal node such that  the edge emanating from a node associated with component $i$ is labeled 1 if and only if $\vb x_i=1$.
This characterization captures common variants of BDDs, including
reduced ordered BDDs~\cite{Wegener_2004}.
When function $\phi$ is given as a binary decision diagram,
SAT, TAUTOLOGY, and NOT-TAUTOLOGY are in P, as they can be solved in linear time by checking the
existence of paths from the root to a terminal node using graph traversal
algorithms~\cite{Wegener_2004}.

\paragraph{As double DNF (DNF01)}
Boolean function $\phi$ can be represented with two DNFs $\phi^0$ and $\phi^1$ of $n$
variables $\vb x_1,\ldots,\vb x_n$ such that $\phi^0$ is satisfied if and only if $\phi(\vb x) =
0$, and $\phi^1$ is satisfied if and only if $\phi(\vb x) = 1$.
This representation is typically employed in Petri nets~\cite{Chaouiya2010,CHJPS14-CMSB} and automata networks~\cite{gored-TCBB}.
In this case,
SAT, TAUTOLOGY, and NOT-TAUTOLOGY are in P from the complexity results on DNFs.

\vspace{2mm}

\begin{example}\label{example:boolean-function-representation}
	We consider propositional formula $f=\vb x_1 \wedge \neg (\vb x_2 \wedge \neg \vb x_3)$ and show its
	different representation schemes and related explanations. 
	\begin{itemize}
		\item $(\vb x_1 \wedge \neg \vb x_2) \vee (\vb x_1\wedge \vb x_3)$
		is an equivalent DNF representation of $f$.
		\item     
		$f$ is unate with   $\preceq$ being $\leq\geq\leq$. On the other hand, another Boolean formula $(\neg \vb x_1\wedge \vb x_2)\vee(\vb x_1\wedge\neg \vb x_2)$ is not unate.
		\item    
		The truth table representation of $f$ is
		$t = 00001101$, assuming that $\operatorname{bin}(1) = 001$. 
		\item     One of the double DNF representations of $f$ is $(\phi^0,\phi^1)$ with
		$\phi^0 = (\neg \vb x_1) \vee (\vb x_2 \wedge \neg \vb x_3)$ and
		$\phi^1 = (\vb x_1 \wedge \neg \vb x_2) \vee (\vb x_1\wedge \vb x_3)$;
		\item An equivalent BDD representation of $f$ is the graph in Figure \ref{fig:example-bdd}:
		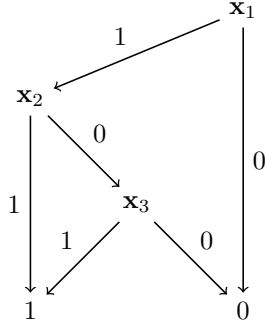
\begin{figure}
			\begin{center}
				\begin{tikzpicture}[->, node distance=2cm, semithick]
					\node (x3)  {${\vb x}_3$};
					\node (one) [below left of=x3]{$1$};
					\node (zero) [below right of=x3] {$0$}; 
					\node (dummy) [above of=zero] {};
					\node (x1) [above of=dummy] {${\vb x}_1$};
					\node (x2) [above left of=x3] {${\vb x}_2$};
					\path[->]
					(x1) edge node [above left] {1} (x2)
					edge node [right] {0} (zero)
					(x2) edge node [left] {1} (one)
					edge node [above right] {0} (x3)
					(x3) edge node [above left] {1} (one)
					edge node [above right] {0} (zero)
					;
				\end{tikzpicture}
				\caption{A binary decidion diagram representation of the Boolean network $f$ in Example \ref{example:boolean-function-representation}}
				\label{fig:example-bdd}
			\end{center}
		\end{figure}
	\end{itemize}
\end{example}

\paragraph{Boolean networks (BNs)}

Recall that a BN of dimension $n$ is  defined by a function $f:\B^n\to \B^n$ with its local (Boolean) function of the $i$-th component $f_i:\B^n\to \B$ for $i\in\range n$.
A BN is \emph{locally monotone} whenever each of its local functions is unate.
In that case, we assume that the orderings of components leading to their unate property are given.
The local Boolean functions of the BN can be encoded with any of the aforementioned representations.
In the case of truth tables, the dimension of the truth table a local function 
follows the number of components which it depends on. A function $f_i$ depends on component $j$ if there exists a configuration $\vb y\in\B^n$ such that $f_i(\vb y_{[1,j-1]} {\vb 0} \vb y_{[j+1,n]}) \neq  f_i(\vb y_{[1,j-1]} {\vb 1} \vb y_{[j+1,n]})$.
In practice, we can significantly reduce the dimension of $f_i$ from $n$ in the following
way. If $k$ is the number of components that $f_i$ depends on, we define its corresponding
integer vector $p\in \range n^k$ to list up the indices of such components in the BN. Then, a truth
table $\TT$ with $2^k$ rows can be constructed to satisfy $f_i(\vb x) = \TT_{\vb x_{p_1}\ldots \vb
	x_{p_k}}$ for any configuration $\vb x\in \B^n$ of the BN.
Finally, a BN can be represented     by its \emph{functional graph},  the digraph of the image by $f$. It is also known as the synchronous state
transition graph. The vertices of such a graph are all the configurations $\B^n$, and there is an edge from $\vb x$ to $\vb y$ if and only if $\vb y = f(\vb x)$. 

\vspace{2mm}

\begin{example} \label{example:locally-monotone} The BN $f:\B^3\to \B^3$ with 
	\begin{align*}
		f_1(\vb x) &= (\neg \vb x_1 \vee \neg \vb x_2) \wedge \vb x_3\\
		f_2(\vb x) &= \vb x_1 \wedge \vb x_3\\
		f_3(\vb x) &= \vb x_1 \vee \vb x_2 \vee \vb x_3
	\end{align*}
	is locally monotone since all its local functions are unate. The functional graph of $f$ is illustrated in Figure \ref{fig:mintrapspace}.
	\begin{figure*}[h]
		\centering
		\begin{tikzpicture}
			\draw[dotted]
			(0,0) node(101) {101}
			-- (3,0) node(111) {111}
			-- (2,1) node(011) {011}
			-- (1,1) node(001) {001}
			-- cycle;
			\draw[dotted]
			(0,3) node(100) {100}
			-- (1,2) node(000) {000}
			-- (2,2) node(010) {010}
			-- (3,3) node(110) {110}
			-- cycle;
			\draw[dotted] (100) -- (101);
			\draw[dotted] (000) -- (001);
			\draw[dotted] (010) -- (011);
			\draw[dotted] (110) -- (111);
			
			\draw[->,black!40] (100) -- (001);
			\draw[->,black!40] (010) -- (001);
			\draw[->,black!40] (001) -- (101);
			\draw[->,black!40] (011) -- (101);
			\draw[->,black!40] (101) -- (111);
			\draw[->,black!40] (111) -- (011);
			\draw[->,black!40] (110) ..  controls +(left:0.5) and +(up:1) ..  (001);
			\draw[->,black!40] (1,2.2) arc (-100:230:2mm);
			\draw[black] (-0.7,-0.3)-- (3.7,-0.3)  -- (2.3,1.3) -- (0.7,1.3) -- cycle;
			\draw[black] (0.7,1.7)-- (1.3,1.7)  -- (1.3,2.3) -- (0.7,2.3) -- cycle;
			\draw[->,black!40] (5, 2) -- (6, 2); 
			\draw[black] (5,0.7) -- (6,0.7)  -- (6,1.3) -- (5,1.3) -- cycle;
			\node[align=left] at (8.3,2) {the synchronous transition};
			\node[align=left] at (7.7,1) {minimal trap space};
		\end{tikzpicture}
		\caption{The functional graph and minimal trap spaces of the Boolean network $f$ in Example \ref{example:locally-monotone}}
		\label{fig:mintrapspace}
	\end{figure*}
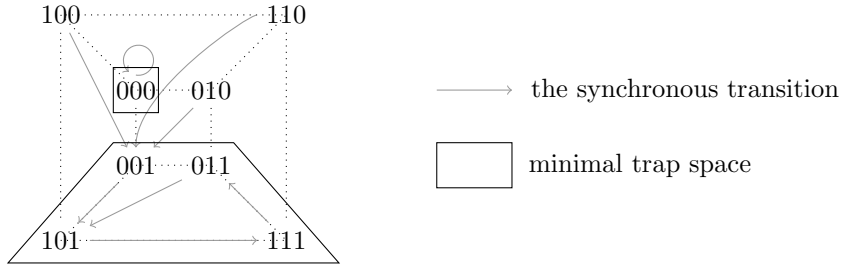
\end{example}
\subsection{Sub-hypercubes and minimal trap spaces of BNs} \label{subsec:notation_trap_space}

A sub-hypercube is an $n$-dimensional sub-graph of the $n$-hypercube such that some dimensions can be fixed to be singular values.
It can be represented as a vector $\vb h\in\C^n$, which specifies for each dimension
$i\in\range n$ whether it is at a fixed value (0 or 1), or free ($*$)  in the sub-hypercube.
The vertices of a sub-hypercube $\vb h$ are denoted by $\vertices(\vb h) := \{\mathbf{x}\in \mathbb{B}^n: \forall i\in [1,n],(\vb h_i \neq * )\Rightarrow (\vb x_i= \vb h_i)\}$.
A sub-hypercube $\vb h$ is smaller than a sub-hypercube space $\vb h'$ whenever 
$\vertices(\vb h)\subseteq\vertices(\vb h')$. We also write this condition as $\vb h\subseteq \vb h'$.

A trap space of a BN $f$ is a sub-hypercube $\vb h\in\C^n$ which is \emph{closed} by $f$, i.e.,
for each vertex $\vb x\in\vertices(\vb h), f(\vb x)\in\vertices(\vb h)$ implying the its image by $f$ is also a vertex of $\vb h$.
Remark that $\vb\ast^n$ is always a trap space.
A trap space $\vb h$ is \emph{minimal} if there is no different trap space $\vb h'\neq \vb h$
within itself; i.e., there exists no trap space $\vb h'$ such that $\vertices(\vb h')\subsetneq \vertices(\vb h)$.
We use $\T(\vb h$) to denote the minimal trap space that contains all configurations in
$\vertices({\vb h})$.
In other words, $\T(\vb h)$ must satisfy three properties:
\begin{itemize}
	\item $\T(\vb h)$ is a trap space, 
	\item $\vb h \subseteq \T(\vb h)$, 
	\item There exists no trap space $\vb h'$ such that  $\vb h \subseteq {\bf h}'
	\subsetneq {\T(\vb h)}$.
\end{itemize}
Remark that if $\vb h$ is a minimal trap space, then, for any configuration $\vb x\in\vertices(\vb
h)$, $\T(\vb x)=\vb h$.

\vspace{2mm}

\begin{example}
	The BN $f$ of {\em Example} \ref{example:locally-monotone} has a fixed point \{000\} and a cyclic attractor \{011, 101, 111\}. 
	See Figure \ref{fig:mintrapspace} for its functional graph representation and minimal trap spaces. 
	It has two minimal trap spaces: $000$ and $\ast\ast1$. 
	Moreover, $\T(010)=\T(01\ast)=\T(0\ast0)=\ast\ast\ast$.
\end{example}

\subsection{Upper bounds results to the computational complexity} \label{subsec:upper-bound_results}
We present all the upper bound results when local functions are given. All the polynomial time solvable cases in Table \ref{tab:results} are also discussed here, except the ones with a functional graph. We later present polynomial time algorithms for the remaining cases in Sec. \ref{sec:hardness-functional-graph}. The basic ideas and previous upper bound results are adopted from \cite{pauleve2020Reconciling}, yet with some extensions to the representations we are considering. All new results are summarized in Theorem \ref{thm:trapspace-tt}.

Consider NOT-TRAPSPACE(\(f,\bf h\)), the problem of deciding if the given hypercube \({\bf h}\) is
\emph{not} closed by \(f\):
it is equivalent to deciding if there exists component
\(i\in\range n\) with \({\bf h}_i\neq *\) and \({\bf z}\in v({\bf h})\) such that \(f_i({\bf z})\neq
{\bf h}_i\),
which boils down to SAT if ${\bf h}_i=0$ and NOT-TAUTOLOGY if ${\bf h}_i=1$.
Thus, the complementary problem TRAPSPACE(\(f,{\bf h}\)) is in coNP for the general case and in P for the locally monotone case. For the same reason, TRAPSPACE(\(f,{\bf h}\)) is in P when the local functions are given as truth tables, BDDs, or double DNFs.

Now, consider NOT-MINTRAP(\(f,\bf h\)), the problem of deciding if the hypercube \({\bf h}\) is either \textit{not} closed by \(f\) or is closed but \textit{not} minimal. 
It can be decided by first checking if \({\bf h}\) is a trap space and then checking the existence of another trap space \({\bf h}'\) which is strictly included in \({\bf h}\). 
This problem is at most NP$^{\text{TRAPSPACE}}$ because only the inclusion ${\bf h}'\subseteq {\bf h}$ needs to be decided  with an oracle for TRAPSPACE, and it can be done in a polynomial time.
Thus, the complementary problem MINTRAP(\(f,{\bf h}\)) is in coNP$^{\text{TRAPSPACE}}$,
which is at most coNP$^{\text{coNP}}=\Pi^{\text P}_2$ in the general case. 
MINTRAP(\(f,{\bf h}\)) is in coNP for the locally monotone case because TRAPSPACE can be solved in a polynomial time. For the same reason,  MINTRAP(\(f,{\bf h}\)) is in coNP when the local functions are represented as truth tables, BDDs, or double DNFs.

Finally, consider IN-MINTRAP$(f,\vb x)$ the problem of deciding whether the configuration $\vb x$ is
a vertex of a minimal trap space of $f$. It boils down to decide MINTRAP($f$, $\T(\vb x)$).
The computation of $\T(\vb x)$ can be performed using the following algorithm:
\begin{lstlisting}
	$\vb h := \vb x$
	repeat $n$ times:
	for each $i\in\range n$ with $\vb h_i\neq \ast$:
	if $\exists \vb y\in\vertices(\vb h)$ s.t. $f_i(\vb y) = 1- \vb y_i$:
	$\vb h_i := \ast$
	return $\vb h$
\end{lstlisting}
The procedure to check the existence in line 4 is equivalent to SAT if $\vb x_i=0$ and NOT-TAUTOLOGY if $\vb x_i=1$.
Thus, overall, this algorithm is in P$^{\text{NP}}$ in the general case, and in P 
for the locally monotone case. Analogously, $\T(\vb x)$ can be computed in a polynomial time when the local functions are represented as truth tables, BDDs, or double DNFs. For all cases, the computational complexity for computing $\T(\vb x)$ does not exceed that of MINTRAP.
Therefore, the computational complexity of IN-MINTRAP is up to the complexity of MINTRAP for each.
\begin{theorem} \label{thm:trapspace-tt}
	Given hypercube $\mathbf{h}$ and BN $f$ with its local functions represented as truth tables, BDDs, or double DNFs,   \emph{TRAPSPACE($f,\mathbf{h}$)}
	can be solved in a polynomial time.
\end{theorem}

\section{Results}\label{sec:main_results}
In this section, we demonstrate computational complexity results for the TRAPSPACE, MINTRAP and IN-MINTRAP
problems in BNs with different representations and the unate property.
In Sec. \ref{sec:hardness-prop-general}, we present the exact computational complexity for BNs with local functions given as propositional formulas, which is the most general case under our consideration. Results for the special case of locally monotone BNs are presented in Sec. \ref{sec:hardness-locally-monotone}.
Those results are used in Sec. \ref{sec:hardness-truth-table} to derive the computational complexity in the case of BNs with local functions represented with truth tables, binary decision diagrams, and double DNFs.
The computational complexity for the BNs given as a functional graph is presented in Sec. \ref{sec:hardness-functional-graph}.

\subsection{Local functions given as propositional formulas\label{sec:hardness-prop-general}}
Theorem \ref{thm:trapspace} demonstrates that TRAPSPACE is coNP-hard when local functions are represented as general propositional formulas, which is the lower bound to the computational complexity. Combined with the previous upper bound results of \cite{pauleve2020Reconciling}, the completeness is shown.

\begin{theorem} \label{thm:trapspace}
	Given hypercube $\mathbf{h}$ and BN $f$ with its local functions represented as propositional formulas,   \emph{TRAPSPACE($f,\mathbf{h}$)}
	is coNP-hard.
\end{theorem}

\begin{proof}
	Consider a Boolean function  $\phi: \mathbb{B}^{n_1} \rightarrow \mathbb{B}$ for $n_1\in \mathbb{Z}^+$ and the associated TAUTOLOGY problem of deciding
	$\forall \mathbf{y}\: \phi(\mathbf{y})=1$, which is  coNP-complete. 
	We construct BN $f: \mathbb{B}^{n_1+1}\rightarrow \mathbb{B}^{n_1+1}$ as
	\begin{align*}
		\forall i\in [1,n_1], \quad
		f_i(\vb x) &= \neg \vb x_i  \\
		f_{n_1+1}(\vb x) & = \phi(\vb x_{[1,n_1]})
	\end{align*}
	and hypercube
	$\mathbf{h}= *^{n_1}{\bf 1}$.
	We prove this theorem by showing that the TAUTOLOGY problem is true if and only if TRAPSPACE($f,\mathbf{h}$) is true.
	If TRAPSPACE($f,\mathbf{h}$) 
	is true, then
	$\forall \mathbf{z}\in   v(\mathbf{h}),  \phi(\vb z_{[1,n_1]})=1$.
	Since $\vb z_{[1,n_1]}$ can have an arbitrary configuration in $\mathbb{B}^{n_1}$, $\forall \mathbf{y} \ \phi(\mathbf{y})$ must be true. 
	On the other hand, if TRAPSPACE($f,\mathbf{h}$)
	is false,  
	we can find configuration ${\vb z} \in \mathbb{B}^{n_1}$ that satisfies $f_{n_1+1}(\vb z{\bf 1})= \phi(\vb z)=0$. 
	This can be used as a certificate that $\forall \mathbf{y}\,   \phi(\mathbf{y})$ is not true, and it can be verified in a polynomial time. Hence, the theorem  holds. 
\end{proof}

Theorem \ref{thm:mintrap-general} and Theorem \ref{thm:inmintrap-general} demonstrate that MINTRAP and IN-MINTRAP are $\mathrm{\Pi}^\text{P}_2$-hard, respectively. 
Combined with the previous upper bound results in \cite{pauleve2020Reconciling}, their completeness is shown.
Our proofs show that MINTRAP and IN-MINTRAP can be used to solve $\mathrm{\Pi}^\text{P}_2$ SAT based on several tricks. First, a component ${\vb x}_i$ with its local function $f_i=\neg {\vb x}_i$ always becomes free in a minimal trap space; see Remark \ref{obs:expansion_negation}. We use this trick to encode Boolean variables quantified with $\forall$ to the BN we construct. Second, given a Boolean formula $\phi$ to be proven its satisfiability, we employ two auxiliary components that have a full control to override other local functions as either 0 or 1 whenever $\phi$ is true. We use this trick to construct a BN that has the full dimensional hypercube as its unique minimal trap space if and only if $\phi$ is satisfied. Note that those auxiliary components will be always presented as the last two components of the BN we construct; see Remarks \ref{obs:auxiliary_vars_expansion}--\ref{obs:vars_expansion_arbitrary} for details.

\begin{remark} \label{obs:expansion_negation}
	Let $f:{\mathbb B}^n \rightarrow {\mathbb B}^{n}$ with $n\in \mathbb{Z}^+$ be a BN. Given $\mathcal{I} \subseteq [1,n]$, suppose $f_{i}(\vb x) = \neg \vb x_{i}$ for $\forall i\in [1,n]$. 
	Then, any hypercube $\mathbf{h} \in \{0,1,*\}^{n}$ must satisfy   ${\bf T(h)}_{i}=*$ for all $i\in \mathcal{I}$.
\end{remark}
\begin{proof}
	For all $i\in \mathcal{I}$,  component $\vb x_i$ can be updated to $\neg \vb x_i$ and realized as both 0 and 1. Therefore, ${\bf T(h)}_{i}=*$ for all $i\in \mathcal{I}$ to ensure that ${\bf T(h)}$ is closed by $f$.
\end{proof}

\begin{remark} \label{obs:auxiliary_vars_expansion}
	For a given Boolean function $\phi: {\mathbb B}^{n} \rightarrow {\mathbb B}$ with $n\in \mathbb{Z}^+$, let $f:{\mathbb B}^{n+2} \rightarrow {\mathbb B}^{n+2}$  be a BN satisfying $\begin{dcases}
		f_{n+1}(\vb x) &=  \phi({\vb x}_{[1,n]}) \wedge \neg \vb x_{n+2}\\
		f_{n+2}(\vb x) &= \vb x_{n+1} \wedge  \neg \vb x_{n+2}
	\end{dcases}$. Suppose hypercube $\mathbf{h} \in \{0,1,*\}^{n+2}$ contains ${\vb z} \in \vertices (\vb h)$ that satisfies $\phi({\vb z}_{[1,n]})=1$. Then, ${\bf T(h)}_{n+1}={\bf T(h)}_{n+2}=*$.
\end{remark}
\begin{proof}
	Consider the transitions of ${\vb z}$ under the fully asynchronous update mode given as follows:
	$$
	{\vb z}_{[1,n]}{\bf 00} \xrightarrow[]{f_{n+1}}
	{\vb z}_{[1,n]}{\bf 10} \xrightarrow[]{f_{n+2}}
	{\vb z}_{[1,n]}{\bf 11} \xrightarrow[]{f_{n+1}}
	{\vb z}_{[1,n]}{\bf 01} \xrightarrow[]{f_{n+2}}
	{\vb z}_{[1,n]}{\bf 00} \xrightarrow[]{f_{n+1}} 
	\cdots,
	$$
	where the label of an arrow corresponds to the local function used to update a configuration. Given any value of ${\vb z}_{[n+1, n+2]}$, 
	components ${\vb x}_{n+1}$ and ${\vb x}_{n+2}$  are  evaluated to be both 0 and 1 during the transition.
	Therefore, ${\bf T(h)}_{n+1}={\bf T(h)}_{n+2}=*$ to ensure that ${\bf T(h)}$ is closed by $f$.
\end{proof}
\begin{remark} 
	\label{obs:vars_expansion_arbitrary}
	Let $f:{\mathbb B}^{n+2} \rightarrow {\mathbb B}^{n+2}$ with $n\in \mathbb{Z}^+$ be a BN. For some $i\in [1,n]$ and a given Boolean function $\phi_i:{\mathbb B}^{n} \rightarrow {\mathbb B}$, suppose the local function $f_{i}(\vb x)$ is in the form of
	$(\phi_i({\vb x}_{[1,n]}) \wedge \neg \vb x_{n+1}) \vee \vb x_{n+2}$. 
	If hypercube $\mathbf{h} \in \{0,1,*\}^{n+2}$ satisfies ${\bf T(h)}_{n+1}={\bf T(h)}_{n+2}=*$,  then ${\bf T(h)}_{i} = *$.
\end{remark}
\begin{proof}
	Since ${\bf T(h)}_{n+1}={\bf T(h)}_{n+2}=*$, there exists a configuration $\vb z \in \vertices({\bf T(h)})$ such that $({\vb z}_{n+1},{\vb z}_{n+2}) = (1,0)$, which can be evaluated as $f_{i}(\vb z)=0$. In addition, there exists another configuration $\vb z' \in \vertices({\bf T(h)})$ such that ${\vb z}'_{n+2} = 1$, which can be evaluated as $f_{i}(\vb z)=1$. Hence, the image of $f_i$ can be both 0 and 1, implying ${\bf T(h)}_{i} = *$ to ensure that ${\bf T(h)}$ is closed by $f$.
\end{proof}

\begin{lemma} \label{lem:propositional_reduction}
	Consider $n_1,n_2\in \mathbb{Z}^+$ with $n_1\leq n_2$ and a Boolean function $\phi: \mathbb{B}^{n_2} \rightarrow \mathbb{B}$ given as a propositional formula. Boolean formula $\forall  \mathbf{y}_{[1,n_1]}
	\exists \mathbf{y}_{[n_1+1,n_2]}
	\phi(\mathbf{y})$ is true if and only if BN $f: \mathbb{B}^{n_2+2} \rightarrow \mathbb{B}^{n_2+2}$ with the local functions defined by \eqref{eq:reduction_propositional_formula_1}--\eqref{eq:reduction_propositional_formula_4} has the unique trap space $*^{n_2+2}$.
	\begin{align}
		\forall j\in [1,n_1], \quad
		f_j(\vb x) &= ({\vb x}_{j} \wedge \neg {\vb x}_{n_2+1}) \vee {\vb x}_{n_2+2}  \label{eq:reduction_propositional_formula_1}
		\\
		\forall j \in [n_1+1, n_2], \quad
		f_{j}(\vb x) &= \neg \vb x_{j}
		\\
		f_{n_2+1}(\vb x) &=  \phi(\vb x_{[1, n_2]}) \wedge \neg \vb x_{n_2+2}
		\\
		f_{n_2+2}(\vb x) &= \vb x_{n_2+1} \wedge  \neg \vb x_{n_2+2} \label{eq:reduction_propositional_formula_4}
	\end{align}
\end{lemma}

\begin{proof}
	If $\forall \vb y_{[1,n_1]} \exists \vb y_{[n_1+1,n_2]} \phi(\vb y)$ is true,  any hypercube $\vb h\in {\{0,1,*\}}^{n_2+2}$ satisfies %
	\begin{align*}
		\mathbf{T}(\vb h) \supseteq &  \mathbf{T}(\vb h_{[1,n_1]} *^{n_2-n_1} \vb h_{[n_2+1,n_2+2]}) && \because Remark~\ref{obs:expansion_negation} \\
		\supseteq &  \mathbf{T}(\vb h_{[1,n_1]} *^{n_2-n_1+2}) && \because Remark~\ref{obs:auxiliary_vars_expansion} \\
		\supseteq & \mathbf{T}(*^{n_2+2}) && \because Remark~\ref{obs:vars_expansion_arbitrary} \\
		= & *^{n_2+2}
	\end{align*}
	Therefore, $*^{n_2+2}$ is the unique (and thus minimal) trap space.
	For the remaining case where   $\exists \vb y_{[1,n_1]} \forall \vb y_{[n_1+1,n_2]} \neg \phi(\vb y)$, 
	$*^{n_2+2}$ is not a minimal trap space because a smaller trap space $\vb h' = \vb y_{[1,n_1]}*^{n_2-n_1}{\bf 0}^2$ exists. This completes the proof.
\end{proof}

\begin{theorem} \label{thm:mintrap-general}
	Given hypercube $\mathbf{h}$ and BN $f$ with its local functions represented as
	propositional formulas, 
	{\rm MINTRAP($f, \mathbf{h}$)} is $\mathrm{\Pi}^\emph{P}_2$-hard.
\end{theorem}
\begin{proof}
	
	Given $n_1,n_2\in \mathbb{Z}^+$ with $n_1\leq n_2$ and a Boolean function $\phi: \mathbb{B}^{n_2} \rightarrow \mathbb{B}$, consider the associated $\Pi_2$SAT problem that deciding whether 
	$\forall  \mathbf{y}_{[1,n_1]}
	\exists \mathbf{y}_{[n_1+1,n_2]}
	\phi(\mathbf{y})$ is true, which is  $\mathrm{\Pi}^\text{P}_2$-complete. 
	By Lemma \ref{lem:propositional_reduction}, this $\Pi_2$SAT problem is true if and only if MINTRAP$\left(f,*^{n_2+2}\right)$ is true for $f$ defined by \eqref{eq:reduction_propositional_formula_1}--\eqref{eq:reduction_propositional_formula_4}. Hence, the theorem holds.
\end{proof}
\begin{theorem} \label{thm:inmintrap-general}
	Given configuration ${\vb x}$ and BN $f$ with its local functions represented as propositional formulas, {\rm IN-MINTRAP($f, {\vb x}$)} is $\mathrm{\Pi}^\emph{P}_2\emph{-hard}$.
\end{theorem}
\begin{proof}
	Given $n_1,n_2\in \mathbb{Z}^+$ with $n_1\leq n_2$ and a Boolean function $\phi: \mathbb{B}^{n_2} \rightarrow \mathbb{B}$, consider the associated $\Pi_2$SAT problem that deciding whether 
	$\forall  \mathbf{y}_{[1,n_1]}
	\exists \mathbf{y}_{[n_1+1,n_2]}
	\phi(\mathbf{y})$ is true, which is  $\mathrm{\Pi}^\text{P}_2$-complete. 
	We prove the theorem by showing that the $\Pi_2$SAT is true if and only if IN-MINTRAP($f,\vb 1^{n_2+2}$) is true for $f$ defined by \eqref{eq:reduction_propositional_formula_1}--\eqref{eq:reduction_propositional_formula_4}.
	
	If $\forall  \mathbf{y}_{[1,n_1]}
	\exists \mathbf{y}_{[n_1+1,n_2]}
	\phi(\mathbf{y})$,  $*^{n_2+2}$ is the unique minimal trap space by Lemma \ref{lem:propositional_reduction} and thus $\vb 1^{n_2+2}$ belongs to a minimal trap space.
	For the remaining case where   $\exists  \mathbf{y}_{[1,n_1]}
	\forall \mathbf{y}_{[n_1+1,n_2]}
	\, \neg\phi(\mathbf{y})$,  we have 
	\begin{align*}
		\mathbf{T}(\vb 1^{n_2+2}) \supseteq  & \mathbf{T}(\vb 1^{n_1} *^{n_2-n_1} \vb 1^{2}) && \because  Remark~\ref{obs:expansion_negation} \\
		\supseteq  &\mathbf{T}(\vb 1^{n_1} *^{n_2-n_1+2}) && \because \vb x_{n_2+2}=1 \\
		\supseteq & \mathbf{T}(*^{n_2+2})  && \because Remark~\ref{obs:vars_expansion_arbitrary} \\
		= & *^{n_2+2}.
	\end{align*}
	However, $\vb *^{n_2+2}$ is not  a minimal trap space because there is a smaller trap space $\vb h' = \vb y_{[1,n_1]}*^{n_2-n_1}\mathbf{0}^2$. This completes the proof.
\end{proof}

\subsection{Locally-monotone BNs with local functions given as propositional formulas}
\label{sec:hardness-locally-monotone}

We show a polynomial-time encoding of any DNF as a BN such that the TAUTOLOGY problem
reduces to MINTRAP and IN-MINTRAP problems. The proofs are given in Theorem \ref{thm:mintrap_locally_monotone} and \ref{thm:inmintrap_locally_monotone}, respectively.
Let us consider any Boolean function $\phi: \mathbb B^n\to \mathbb B$ represented in DNF as a list of
$k$ conjunctive clauses.
For $j\in [1,k]$, we use $c_j(\vb y)$ to denote the $j$-th clause of $\phi$ evaluated with $\vb y \in \mathbb{B}^n$.
Whenever $k=0$, $\phi$ is considered to be false. Whenever a clause is empty it is
equivalent to be true. 
We can assume that each clause $c_j(\vb y)$ does not contain a contradiction caused by the same component (e.g., $\vb y_i\land \neg \vb y_i$). Therefore, all clauses are unate.

\begin{lemma} \label{lem:locally-monotone}
	Let us consider $n, k\in \mathbb{Z}^+$ and a Boolean function $\phi: \mathbb{B}^n \rightarrow \mathbb{B}$ given as a DNF with $k$ conjunctive clauses; i.e., $\phi(\vb y):= \bigvee_{j=1}^{k} c_j(\vb y)$. Boolean formula $\forall \vb y \phi(\vb y)$ is true if and only if BN $f: \mathbb B^{n+k+2}\to \mathbb B^{n+k+2}$ with the local functions defined by \eqref{eq:reduction_locally_monotone_1}--\eqref{eq:reduction_locally_monotone_4} has the unique minimal trap space $*^{n+k+2}$.
	\begin{align}
		\forall i\in \range n,\quad
		f_i(\vb x) &= (\vb x_i \wedge \neg \vb x_{n+k+1}) \vee \vb x_{n+k+2} \label{eq:reduction_locally_monotone_1}
		\\
		\forall j\in \range k,\quad
		f_{n+j}(\vb x) &= (c_j(\vb x_{[1,n]}) \wedge \neg \vb x_{n+k+1}) \vee \vb x_{n+k+2} \label{eq:reduction_locally_monotone_2}
		\\
		f_{n+k+1}(\vb x) &= \left(\textstyle\bigvee_{j=1}^k \vb x_{n+j}\right) \wedge \neg \vb x_{n+k+2} \label{eq:reduction_locally_monotone_3}
		\\
		f_{n+k+2}(\vb x) &= \vb x_{n+k+1}\wedge \neg \vb x_{n+k+2} \label{eq:reduction_locally_monotone_4}
	\end{align}
\end{lemma}
\begin{proof}
	If $\forall \vb y \, \phi(\vb y)$ is true,  any hypercube $\vb h\in {\{0,1,*\}}^{n+k+2}$ satisfies $\mathbf{T}(\vb h)=*^{n+k+2}$ by the  \textbf{Case (i)} and \textbf{Case (ii)}. 
	\begin{description}
		\item[Case (i)]: When  $\vb h_{n+k+1}=0$,~\\
		Eq. \eqref{eq:reduction_locally_monotone_2} can be simplified to $c_j(\vb x_{[1,n]}) \vee \vb x_{n+k+2}$.
		For an arbitrary element $\vb z \in v(\vb h)$, we can find $j^*\in [1,k]$ such that $c_{j^*}({\vb x}_{[1,n]})=1$ since $\forall \vb y \, \phi(\vb y)$ is true. Consequently,
		\begin{align*}
			\mathbf{T}(\vb h) \supseteq &  \mathbf{T}(\vb h_{[1,n+j^*-1]} \vb 1 \vb h_{[n+j^*+1,n+k+2]}) && \because {\vb x}_{n+j^*} \text{ can be evaluated to be 1}\\
			\supseteq &  \mathbf{T}(\vb h_{[1,n+j^*-1]} \vb 1 \vb h_{[n+j^*+1,n+k]}*^{2}) && \because Remark~\ref{obs:auxiliary_vars_expansion} \\
			\supseteq & \mathbf{T}(*^{n+k+2}) && \because Remark~\ref{obs:vars_expansion_arbitrary} \\
			= & *^{n+k+2}.
		\end{align*}
		\item[Case (ii)]: When  $\vb h_{n+k+1}\in \{1,*\}$,~\\
		Eq. \eqref{eq:reduction_locally_monotone_4} simplifies to $\neg {\vb x}_{n+k+2}$. Consequently,
		\begin{align*}
			\mathbf{T}(\vb h) \supseteq &  \mathbf{T}(\vb h_{[1,n+k+1]} *) && \because Remark~\ref{obs:expansion_negation} \\
			\supseteq &  \mathbf{T}(\vb h_{[1,n]} {\vb 1}^k \vb h_{n+k+1} *) && \because {\vb x}_{n+k+2} \text{ can be evaluated to be 1}\\
			\supseteq &  \mathbf{T}(\vb h_{[1,n]} {\vb 1}^k *^{2}) && \because Remark~\ref{obs:auxiliary_vars_expansion}\\
			\supseteq & \mathbf{T}(*^{n+k+2}) && \because Remark~\ref{obs:vars_expansion_arbitrary} \\
			= & *^{n+k+2}.
		\end{align*}
	\end{description}
	Therefore, $*^{n+k+2}$ is the unique minimal trap space if $\forall \vb y \, \phi(\vb y)$ is true. 
	On the other hand, if $\exists \vb y \, \neg \phi(\vb y)$ is true, then $*^{n+k+2}$ is not a minimal trap space because there is a smaller trap space $\vb h' = \vb y \mathbf{0}^{k+2}$. Hence the lemma holds.
\end{proof}

\begin{theorem} \label{thm:mintrap_locally_monotone}
	Given hypercube $\mathbf{h}$ and locally-monotone BN $f$ with local functions represented as propositional formulas, {\rm MINTRAP($f, \mathbf{h}$)} is coNP-hard.
\end{theorem}
\begin{proof}
	Given $n\in \mathbb{Z}^+$ and a Boolean function $\phi: \mathbb{B}^n \rightarrow \mathbb{B}$ in a DNF with $k$ conjunctive clauses;  i.e., $\phi(\vb y):=\bigvee_{j=1}^{k} c_j(\vb y)$,  consider the associated TAUTOLOGY problem  $\forall \vb y \, \phi(\vb y)$, which is coNP-complete.  By Lemma \ref{lem:locally-monotone}, the TAUTOLOGY problem is true if and only if MINTRAP($f$,$*^{n+k+2}$) is true for $f$ 
	defined by \eqref{eq:reduction_locally_monotone_1}--\eqref{eq:reduction_locally_monotone_4}. Since all local functions are unate, the theorem holds.
\end{proof}

\begin{theorem} \label{thm:inmintrap_locally_monotone}
	Given configuration $\vb x$ and locally-monotone BN $f$ with local functions represented as propositional formulas, {\rm IN-MINTRAP($f, \vb x$)} is coNP-hard.
\end{theorem}
\begin{proof}
	Suppose $n\in \mathbb{Z}^+$ and a Boolean function $\phi: \mathbb{B}^n \rightarrow \mathbb{B}$ in a DNF with $k$ conjunctive clauses are given (i.e., $\phi(\vb y):=\bigvee_{j=1}^{k} c_j(\vb y)$). Consider the associated TAUTOLOGY problem  $\forall \vb y \, \phi(\vb y)$, which is  coNP-complete. We prove the theorem by showing that the TAUTOLOGY problem is true if and only if IN-MINTRAP($f$,$\vb 1^{n+k+2}$) is true for $f$ defined by \eqref{eq:reduction_locally_monotone_1}--\eqref{eq:reduction_locally_monotone_4}.
	
	If $\forall \vb y \, \phi(\vb y)$ is true, then $*^{n+k+2}$ is the unique trap space by Lemma \ref{lem:locally-monotone} and thus $\vb 1^{n+k+2}$ belongs to a minimal trap space. For the remaining case where $\exists \vb y \, \neg \phi(\vb y)$ is true, 
	\begin{align*}
		\mathbf{T}(\vb 1^{n+k+2}) \supseteq  & \mathbf{T}(\vb 1^{n+k} *^{2}) && \because \vb x_{n+k+2}=1 \\
		\supseteq & \mathbf{T}(*^{n+k+2})  && \because Remark~\ref{obs:vars_expansion_arbitrary} \\
		= & *^{n+k+2}.
	\end{align*}
	However, $\vb *^{n+k+2}$ is not  a minimal trap space because there is a smaller trap space $\vb h' = \vb y\mathbf{0}^{k+2}$. This completes the proof.
\end{proof}

\subsection{With local functions represented as truth tables, BDDs, and double DNFs\label{sec:hardness-truth-table}} 

We now consider any BN whose local functions are represented either as truth tables, BDDs, or double DNFs.
In Theorem \ref{thm:tt}, we prove the lower bound results to the computational complexity of MINTRAP and IN-MINTRAP problems by reduction of TAUTOLOGY to 3DNF. Combined with the  upper bound results presented in Sec. \ref{subsec:upper-bound_results}, the completeness is shown.

Consider the encoding of the clauses $\phi$ as the BN $f$ defined by \eqref{eq:reduction_locally_monotone_1}--\eqref{eq:reduction_locally_monotone_4}.
Remark that all local functions but  $f_{n+k+1}$ in \eqref{eq:reduction_locally_monotone_3} depend on at most 5 variables, and thus each of them can be encoded in constant space and time as a truth table, a BDD, or double DNFs.
However, the local function $f_{n+k+1}$ in \eqref{eq:reduction_locally_monotone_3} depends on $k+1$ variables, where $k$ is the number of clauses in the
DNFs. Therefore, converting this local function may require an exponential time and space. We resolve this issue by appending a small number of auxiliary variables  that correspond to local functions having a constant size. Note that \eqref{eq:reduction_locally_monotone_3} is true whenever at least one of the clauses can be evaluated to be true and the component $\vb x_{n+k+2}$ is false.
This definition can be incorporated by appending $k$ additional components with at most two literals to the BN so that the $j$-th element of the new components is evaluated to be true if either 
$c_j$ or the $(j-1)$-th element of the new components can be true for $j \in [1,k]$. As a consequence, the $k$-th of the additional components is true whenever at least one clause of $\phi$ can be evaluated to be true.
We adapt this idea by expanding locally monotone BN \eqref{eq:reduction_locally_monotone_1}--\eqref{eq:reduction_locally_monotone_4} to  \eqref{eq:reduction_tt_1}--\eqref{eq:reduction_tt_6}, which can be encoded in constant space and time as truth table, BDD, or double DNFs. We employ Remark \ref{obs:tt_expansion} and Lemma \ref{lem:truth-table} to prove Theorem \ref{thm:tt}.
\begin{align}
	\forall i\in \range n\quad
	f_i(\vb x) &= (\vb x_i \wedge \neg \vb x_{n+k+1}) \vee \vb x_{n+k+2} \label{eq:reduction_tt_1}
	\\ 
	\forall j\in \range k,\quad
	f_{n+j}(\vb x) &= (c_j(\vb x_{[1,n]}) \wedge \neg \vb x_{n+k+1}) \vee \vb x_{n+k+2} \label{eq:reduction_tt_2}
	\\ 
	f_{n+k+1}(\vb x) &=  \vb x_{n+1} \label{eq:reduction_tt_3}\\
	\forall j\in \frange 2 k,\quad
	f_{n+k+j}(\vb x) &=  \vb x_{n+j} \vee \vb x_{n+k+j-1} \label{eq:reduction_tt_4}\\
	f_{n+2k+1}(\vb x) &= \vb x_{n+2k} \wedge \neg \vb x_{n+2k+2} \label{eq:reduction_tt_5}
	\\
	f_{n+2k+2}(\vb x) &= \vb x_{n+2k+1}\wedge \neg \vb x_{n+2k+2} \label{eq:reduction_tt_6}
\end{align}

\begin{remark} \label{obs:tt_expansion} 
	Consider a BN $f$ given as \eqref{eq:reduction_tt_1}--\eqref{eq:reduction_tt_6}.
	If $\mathbf{T}(\vb h)_i = *$ for all $i \in [n+1,n+k]$, we can sequentially show that $\mathbf{T}(\vb h)_{n+k+j} = *$ by increasing $j$ from 1 to $k$. This is because  $\mathbf{T}(\vb h)_{n+j}$ and $\mathbf{T}(\vb h)_{n+k+j-1}$ are both $*$ and thus Eq.\eqref{eq:reduction_tt_4} can be evaluated to be both 0 and 1.\\
\end{remark}

\begin{lemma} \label{lem:truth-table}
	Consider $n\in \mathbb{Z}^+$ and a Boolean function $\phi: \mathbb{B}^n \rightarrow \mathbb{B}$ given as a 3DNF with $k$ conjunctive clauses that contain at most three literals, i.e., $\phi(\vb y):= \bigvee_{j=1}^{k} c_j(\vb y)$. Boolean formula $\forall \vb y \phi(\vb y)$ is true if and only if BN $f: \mathbb B^{n+2k+2}\to \mathbb B^{n+2k+2}$ with the local functions defined by \eqref{eq:reduction_tt_1}--\eqref{eq:reduction_tt_6} has the unique minimal trap space $*^{n+2k+2}$.
\end{lemma}

\begin{proof}
	If $\forall \vb y \, \phi(\vb y)$ is true,  any hypercube $\vb h\in {\{0,1,*\}}^{n+2k+2}$ satisfies $\mathbf{T}(\vb h)=*^{n+2k+2}$ by the  \textbf{Case (i)} and \textbf{Case (ii)}. 
	\begin{description}
		\item[Case (i)]: When  $\vb h_{n+2k+1}=0$,~\\
		Eq. \eqref{eq:reduction_tt_2} can be simplified to $c_j(\vb x_{[1,n]}) \vee \vb x_{n+2k+2}$ because every configuration satisfies $\vb  x_{n+2k+1}=0$.
		For an arbitrary element $\vb z \in v(\vb h)$, we can find $j^*\in [1,k]$ such that $c_{j^*}
		({\vb x}_{[1,n]})=1$ since $\forall \vb y \, \phi(\vb y)$ is true. Therefore, $\mathbf{T}(\vb h)_{n+j^*} \in \{1,*\}$ and subsequently, 
		\begin{align*}
			\mathbf{T}(\vb h) \supseteq &  \mathbf{T}(\vb h_{[1,n+k+j^*-1]} \vb 1 \vb h_{[n+k+j^*+1,n+2k+2]}) && \because {\vb x}_{n+k+j^*}  \text{ can be evaluated to be 1}\\
			\supseteq &  \mathbf{T}(\vb h_{[1,n+k+j^*-1]} \vb 1^{(k-j^*+1)} \vb h_{[n+2k+1,n+2k+2]})
			&& \because \text{Increasing } j \text{ from } (j^*+1) \text{ to } k, {\vb x}_{n+k+j} \\
			&&& \quad \text{can be sequentially evaluated to be 1}\\
			\supseteq &  \mathbf{T}(\vb h_{[1,n+k+j^*-1]} \vb 1^{(k-j^*+1)}*^{2}) && \because Remark~\ref{obs:auxiliary_vars_expansion} \\
			\supseteq &  \mathbf{T}(\vb *^{n+k} \vb h_{[n+k+1,n+k+j^*-1]} \vb 1^{(k-j^*+1)}*^{2}) && \because Remark~\ref{obs:vars_expansion_arbitrary} \\
			\supseteq & \mathbf{T}(*^{n+2k+2}) && \because Remark~\ref{obs:tt_expansion}\\
			= & *^{n+2k+2}.
		\end{align*}
		\item[Case (ii)]: When  $\vb h_{n+2k+1}\in \{1,*\}$,~\\
		Eq. \eqref{eq:reduction_tt_6} simplifies to $\neg {\vb x}_{n+2k+2}$. Consequently,
		\begin{align*}
			\mathbf{T}(\vb h) \supseteq & \mathbf{T}(\vb h_{[1,n+2k+1]} *) && \because Remark~\ref{obs:expansion_negation} \\
			\supseteq &  \mathbf{T}(\vb h_{[1,n]} {\vb 1}^k \vb h_{[n+k+1,n+2k+1]} *) && \because {\vb x}_{n+2k+2} \text{ can be evaluated to be 1}\\
			\supseteq &  \mathbf{T}(\vb h_{[1,n]} {\vb 1}^{2k} \vb h_{n+2k+1} *) && \because \forall j\in [1,k], {\vb x}_{n+k+j} \text{ can be evaluated to be 1 by } {\vb x}_{n+j}=1\\
			\supseteq &  \mathbf{T}(\vb h_{[1,n]} {\vb 1}^{2k} *^{2}) && \because Remark~\ref{obs:auxiliary_vars_expansion}\\
			\supseteq & \mathbf{T}(*^{n+k}{\vb 1}^{k} *^{2}) && \because Remark~\ref{obs:vars_expansion_arbitrary} \\
			\supseteq & \mathbf{T}(*^{n+2k+2}) &&  \because Remark~\ref{obs:tt_expansion}\\
			= & *^{n+2k+2}.
		\end{align*}
	\end{description}
	Therefore, $*^{n+2k+2}$ is the unique minimal trap space. 
	On the other hand, if $\exists \vb y \, \neg \phi(\vb y)$ is true, $*^{n+2k+2}$ is not a minimal trap space because there is a smaller trap space $\vb h' = \vb y \mathbf{0}^{2k+2}$. Hence the lemma holds.
\end{proof}

\begin{theorem}\label{thm:tt}
	MINTRAP and IN-MINTRAP are coNP-hard for BNs with local functions represented with
	truth tables, binary decision diagrams, and double DNFs.
\end{theorem}
\begin{proof}
	The local functions defined by \eqref{eq:reduction_tt_1}--\eqref{eq:reduction_tt_6} can be encoded in a polynomial time as truth tables, BDDs,
	or double DNFs. Therefore, the theorem holds by Lemma \ref{lem:truth-table}.
\end{proof}

\subsection{Functional graphs of BNs \label{sec:hardness-functional-graph}}

Now consider the case when
the BN $f:\mathbb B^n \rightarrow \mathbb B^n$ is represented by its functional digraph $G = (V,E)$
with $V=\B^n$ and $E = \{ (\vb x, f(\vb x)) \mid \vb x\in\B^n\}$.
Given a vertex $\vb x\in V$, we write $out(\vb x) = \{ \vb y\mid (\vb x,\vb y)\in E\}$. 
Note that in the case of the functional graph, $out(\vb x) = \{ f(\vb x)\}$, which is a singleton set.
Given a set of vertices $V'\subseteq V$, we can consider a subgraph $G_{V'}=(V',\{(u,w)\in E\mid
u\in V', w\in V'\})$.

For a given sub-hypercube $\vb h\in\C^n$ to be a trap space, each $\vb x\in \vertices(\vb h)$ must
verify that $out(\vb x)\subseteq \vertices(\vb h)$.
Therefore, TRAPSPACE can be solved in time linear to the size of $G$ (number of vertices plus edges, $|V|+|E|$).

Our algorithm for the decision of MINTRAP uses two auxiliary functions \texttt{SUB-HYPERCUBE} and \texttt{SATURATE}.
Function \texttt{SUB-HYPERCUBE} returns the smallest
enclosing sub-hypercube for a given a non-empty sublist of vertices $W\subseteq V$. Function \texttt{SATURATE}  computes
$\T(\texttt{SUB-HYPERCUBE}(W))$ for a given non-empty sublist of vertices $W\subseteq V$. In other words, \texttt{SATURATE} computes the smallest sub-hypercube that encloses $W$ and  is
closed by $f$.

\noindent
\texttt{SUB-HYPERCUBE}($W:=(W_1, \ldots)$):
\begin{lstlisting}
$\vb h$ = $W_1$
for each $\vb x\in W$:
	for $i\in\range n$:
        if $\vb h_i\in\B$ and $\vb h_i\neq \vb x_i$:
            $\vb h_i := \ast$
return $\vb h$
\end{lstlisting}

\noindent
\texttt{SATURATE}($W$):
\begin{lstlisting}
$\vb h := $ SUB-HYPERCUBE($W$)
repeat
	$\vb h' := \vb h$
	$W := \vertices(\vb h) \cup \bigcup_{u\in \vertices(\vb h)} out(u)$
	$\vb h := $ SUB-HYPERCUBE($W$)
until $\vb h = \vb h'$
return $\vb h$
\end{lstlisting}
Remark that \texttt{SATURATE} runs in a polynomial time to the size of $G$ as the the loop in line 2-6
is performed at most $n$
times.

One can decide whether the sub-hypercube $\vb h$ is a minimal trap space by computing the terminal strongly connected components of $G$ which are enclosed in $\vb h$ and verify that their smallest enclosing trap space is $\vb h$.
Indeed, consider that $\vb h$ is a trap space. By definition, the saturation of any set of its vertices gives a trap spaces which is either equal to or smaller than $\vb h$.
Then, remark that any trap space within $\vb h$ contains at least one terminal strongly connected component of $G_{\vertices(\vb h)}$.
Therefore, it is sufficient to verify that the saturation of all these terminal strongly connected components are not strictly smaller than $\vb h$ to determine that $\vb h$ is minimal.

We call the algorithm computing the
terminal strongly connected components \texttt{terminal-SCCs} and it can be done in a
polynomial time to the size of $G$ (e.g., with Tarjan's algorithm~\cite{Tarjan72}). 

\noindent
\texttt{IS\_MINTRAP}($G$, $\vb h$):
\begin{lstlisting}
if not TRAPSPACE($G$, $\vb h$):
	return False
tSCCs := terminal-SCCs($G_{\vertices(\vb h)}$)
for each $W$ in tSCCs:
	if SATURATE($W$) $\neq$ $\vb h$:
        return False
return True
\end{lstlisting}
This algorithm runs in a polynomial time to the size of $G$.
Finally, remark that IN-MINTRAP($f$, $\vb x$) can be decided using \texttt{IS\_MINTRAP}($G$, \texttt{SATURATE}($\{\vb x\}$)), which  also runs in a polynomial time to the size of $G$.

\begin{theorem}
	TRAPSPACE, MINTRAP, and IN-MINTRAP are in P for BNs 
	given  as their functional graph.
\end{theorem}
The functional graph of $f$ corresponds to the so-called state transition graph with the synchronous
(parallel) update mode: each edge corresponds to a synchronous transition.
One can remark that the above algorithms give equivalent results with the fully asynchronous state
transition graph where $out'(\vb x) = \{ \vb y \in\B^n\mid \exists i\in\range n, \vb y_i=f_i(\vb x),
\forall j\in\range n, j\neq i, \vb x_j=\vb y_j\}$.
Indeed, \texttt{SUB-HYPERCUBE}($\{\vb x, f(\vb x)\}$) is always equal to  \texttt{SUB-HYPERCUBE}($\{\vb x\}\cup
out'(\vb x)$); remark that, for any $i\in\range n$,
$f_i(\vb x) \neq \vb x_i$ if and only if
there exists $\vb y\in out'(\vb x)$ such that $\vb y_i\neq \vb x_i$.

\section{Conclusion} \label{sec:conclusion}
In this paper, we characterized the computational complexity of three important decision problems related to trap spaces in Boolean networks considering various representations and the locally monotone case.
We demonstrated that, in general, determining minimal trap space properties and the membership of
configurations to minimal trap spaces are equivalent to solving the satisfiability of 
Boolean formulas with two alternating quantifiers $\forall$ and $\exists$. Hence, our results show that they are $\Pi_2^{\text P}$-complete.
However, whenever restricting to the cases whenever BN is locally monotone, or whenever its local
functions are encoded as truth tables, binary decision diagrams, or double DNFs (such as Petri nets
encodings of BNs), the complexity drops by one level in the polynomial hierarchy and becomes
equivalent to the decision of the tautology property of propositional formulas.
Finally, whenever the BN is given by its functional graph (corresponding to its synchronous state
transition graph), minimal trap space properties can be decided by deterministic algorithms in
a polynomial time.

In practice, solving coNP problems can be tackled with SAT solvers, whereas solving $\Pi_2^{\text P}$ necessitates more elaborated approaches, such as Answer-Set Programming~\cite{Eiter1995} or by decomposing the problem into two parts and alternately solving them, as demonstrated by a recent study on the control of fixed points \cite{moon2022Bilevel}.

Future direction may consider studying the computational complexity of problems related to the set
of minimal trap spaces of a BN, such as deciding whether all the minimal trap spaces satisfy a given
property. This will give insight into the complexity for control problems related to minimal trap
spaces in BNs, as tackled in~\cite{MarkerReprogramming,Rozum2021}.

\bibliographystyle{siamplain}
\bibliography{reference}

\end{document}

%% file: tmp_main_header.tex
\title{Computational Complexity of Minimal Trap Spaces in Boolean Networks\thanks{%
\funding{KM and KL acknowledges support from the National Research Foundation of Korea(NRF) grant
	funded by the Korea government(MEST) (grant number: NRF-2022R1F1A1074140)
	LP acknowledges support from the French Agence Nationale pour la Recherche (ANR) in the scope of the project BNeDiction (grant number: ANR-20-CE45-0001).}}}

\author{Kyungduk Moon\thanks{Department of Industrial and Management Engineering, POSTECH, South Korea (\email{kaleb.moon@postech.ac.kr}, \email{kblee@postech.ac.kr}).}
\and Kangbok Lee\footnotemark[2]
\and Loïc Paulevé\thanks{Univ. Bordeaux, CNRS, Bordeaux INP, LaBRI, UMR 5800 F-33400 Talence, France (\email{loic.pauleve@labri.fr})}}

\headers{Computational complexity of trap spaces in Boolean networks}{Kyungduk Moon, Kangbok Lee, and Loïc Paulevé}

%% file: tmp_main_abstract.tex
\begin{abstract}
A Boolean network (BN) is a discrete dynamical system defined by a Boolean function that maps to the domain itself. A trap space of a BN is a generalization of a fixed point, which is defined as the sub-hypercubes closed by the function of the BN. A trap space is minimal if it does not contain any smaller trap space.
Minimal trap spaces have applications for the analysis of attractors of BNs with various update modes.
This paper establishes the computational complexity results of three decision problems related to minimal trap spaces:
the decision of the trap space property of a sub-hypercube,
the decision of its minimality, and
the decision of the membership of a given configuration to a minimal trap space.
Under several cases on Boolean function
representations, we investigate the computational complexity
of each problem.
In the general case, we demonstrate that the trap space property is coNP-complete, and the minimality
and the membership properties are $\Pi_2^{\text P}$-complete.
The complexities drop by one level in the polynomial hierarchy whenever the local functions of the
BN are either unate, or are
represented using truth-tables, binary decision diagrams, or double DNFs (Petri net encoding):
the trap space property can be decided in a polynomial time, whereas deciding the minimality and the membership are coNP-complete.
When the BN is given as its functional graph, all these problems are in P.
\end{abstract}

\begin{keywords}
Automata network, Trap space, Computational complexity, Boolean function representation, System dynamics, Attractors
\end{keywords}

\begin{MSCcodes}
68Q17, 68R07, 94C11, 37M22, 37N25
\end{MSCcodes}